\documentclass[runningheads]{llncs}
\usepackage[T1]{fontenc}
\usepackage{graphicx}
\usepackage[english]{babel}
\usepackage[utf8]{inputenc}
\usepackage{times}
\usepackage{soul}
\usepackage{amsmath}
\usepackage{amssymb}
\usepackage{url}
\usepackage{booktabs}
\usepackage{utfsym}
\usepackage{wasysym}
\usepackage{makecell}
\usepackage{algorithm}
\usepackage{algpseudocode}
\usepackage{pifont}
\newcommand{\cmark}{\ding{51}} 
\newcommand{\xmark}{\ding{55}} 

\usepackage{graphicx}
\usepackage[colorlinks=true, allcolors=blue]{hyperref}
\usepackage[small]{caption}
\usepackage[switch]{lineno}
\usepackage{xcolor}
\usepackage{dblfloatfix}
\usepackage{authblk}
\usepackage{tikz}
\usetikzlibrary{positioning}
\usepackage{caption}
\usepackage[numbers,sort&compress]{natbib}
\usepackage{comment}
\usepackage{color}

\begin{document}
\title{Temporal Fair Division of Indivisible Goods with Structured Constraints}

\author{Kui-Wang Choi\inst{1}\orcidID{0009-0009-8723-0223} \and
Minming Li\inst{2}\orcidID{0000-0002-7370-6237}}

\institute{Department of Computer Science, City University of Hong Kong 
\email{kuiwchoi2-c@my.cityu.edu.hk}\\
\and
Department of Computer Science. City University of Hong Kong
\email{minming.li@cityu.edu.hk}}

\maketitle              
\begin{abstract}
This paper investigates temporal fair division, a setting where items are allocated over multiple rounds and agents require cumulative fairness over time. We focus on dynamic extensions of classic fairness notions: Temporal Envy-Freeness Up to Any Good (TEFX), its $\alpha$-TEFX approximation, and Temporal Maximin Share (TMMS). Because these strict fairness criteria are known to be generally impossible to satisfy, we analyze the model under constraints to map the boundary between what is possible and what is not.

Our contribution systematically maps the structural boundaries of these temporal fairness notions. Our main technical result introduces a novel dynamic backtracking framework that achieves exact TEFX for Strong Binary Valuations. We prove, via a lexicographical potential function argument, that this bounded historical reallocation systematically resolves temporal envy cycles and terminates in finite time. Finally, we establish tight approximation ratios for $\alpha$-TEFX under identical valuations and bi-valued goods, alongside targeted impossibility results for TEFX and TMMS, demonstrating exactly where temporal fairness is mathematically unattainable, and properly contextualizing prior algorithms that only implicitly satisfied our newly defined TMMS metric.

\end{abstract}

\section{Introduction}
Fair division is a fundamental problem in multi-agent systems and economics, traditionally focusing on the static allocation of resources, such as splitting a cake or distributing a set of goods at a single point in time \citep{lipton2004,budish2011,gourves2014,caragiannis2019,plaut2017,aziz2021,chan2019,amanatidis2020}. However, many real-world allocation problems are dynamic. For example, in food-bank donation, goods often arrive over time and must be allocated online \citep{aleksandrov2015}. Although this approach captures the dynamic feature for real-world problems, it ignores the fairness in intermediate allocations. Therefore, the goal of temporal fair division is to maintain fairness cumulatively in each round, rather than just over the long run. This is crucial in some real-world allocation problems, such as in the food bank donation problem, the allocators receive a timeline for the donation in order to plan ahead, such that the receivers can be as fair as possible during the period of donation.

The transition from dynamic to temporal settings introduces significant theoretical challenges. Unlike dynamic fair division, where intermediate allocations need not be fair, temporal mechanisms must account for this. Recent work has painted a bleak picture for this domain: temporal fair allocations are often impossible to guarantee without compromising efficiency or social welfare \citep{he2019,elkind2024,cookson2024}. For instance, strong fairness concepts such as Envy-Freeness up to any good (EFX) are often unattainable in temporal settings because an agent's perception of fairness fluctuates wildly as new goods arrive.

Despite these impossibility results, the literature leaves several critical avenues unexplored. Existing impossibilities heavily rely on broad valuation classes. It is unclear whether more structured preferences, such as identical valuations or strong binary valuations, or structured arrival goods, such as identical days, can circumvent these temporal frictions. In addition, when exact temporal fairness is mathematically infeasible, the precise limits of approximate temporal fairness ($\alpha$-TEFX) and alternative fairness notions, such as Temporal Maximin Share (TMMS), have not been formally mapped.

In this work, we directly address these gaps. We focus on three fairness metrics that have not been considered by previous works:
\begin{itemize}
    \item \textbf{TEFX}: Temporal Envy-freeness Up to Any Good
    \item \textbf{$\alpha$-TEFX}: Approximation of TEFX
    \item \textbf{TMMS}: Temporal Maximin Share
\end{itemize}

For each fairness metric, we seek the boundaries on four structured constraints:
\begin{itemize}
    \item \textbf{Identical Days}: The goods that arrive in each round are identical.
    \item \textbf{Strong Binary Valuation}: Each agent values each good either 0 or 1, and no good is valued 0 for every agent.
    \item \textbf{Identical Valuation}: Each agent's valuation profile is identical.
    \item \textbf{Bi-valued Goods}: Each agent values each good either $a$ or $b$, where $a$ and $b$ are positively constants.
\end{itemize}

Therefore, our main research question is to \textit{examine the boundaries of temporal fairness under various structured constraints in temporal fair division}.

\subsection{Our Contribution}
We summarize our detailed results and position them alongside prior work on temporal fair division in Table~\ref{tab:tef}.

In Section~\ref{sec:strongbinval}, we present our key contribution: a novel algorithm that returns a TEFX allocation under the strong binary valuation setting.

In Section~\ref{sec:identicaldays}, we establish a negative result for TEFX and TMMS under the identical days setting, while showing that a constant $1/2$ approximation ratio exists for TEFX when there are two agents.

In Section~\ref{sec:identicalval} and Section~\ref{sec:bival}, we build on the negative result of TEFX \citep{elkind2024} under the identical valuation and bi-valued goods setting, providing two algorithms that return a tight approximation ratio for TEFX, respectively. We also construct a counterexample to show that a TMMS allocation may not exist even under both settings simultaneously.

\begin{table}
  \centering
  \caption{Possibilities, impossibilities, and open questions in temporal fair division under various constraints. \cmark indicates a possibility result; \xmark indicates an impossibility; \(\dagger\) indicates a tight approximation ratio}
  \label{tab:tef}
    \begin{tabular}{|c|c|c|c|c|}
      \hline
       & TEFX & \(\alpha\)-TEFX & TMMS \\
      \hline
      \makecell{Strong Binary \\ Valuation} & \makecell{\cmark \\ (\textcolor{red}{Theorem~\ref{thm:tefxbinval}})} & \makecell{1} & \makecell{\cmark \\ (\textcolor{red}{Corollary~\ref{cor:tmms_existence}})} \\
      \hline
      Identical Days & \makecell{\xmark\\(\textcolor{red}{Theorem~\ref{thm:tefxidenticaldays}})} & \makecell{\(\frac{1}{2}\) for \(n = 2\) \\ (\textcolor{red}{Theorem~\ref{thm:atefxidenticaldays}})} & \makecell{\xmark \\ (\textcolor{red}{Theorem~\ref{thm:tmmsidenticaldays}})} \\
      \hline
      Identical Valuation & \makecell{\xmark \\ \citep{elkind2024}} & \makecell{\(\frac{\min_{g \in O} v(g)}{\max_{g \in O} v(g)}\) \(\dagger\) \\ (\textcolor{red}{Theorem~\ref{thm:atefxidenticalvaluation}})} & \makecell{\xmark \\ (\textcolor{red}{Theorem~\ref{thm:tmmsidenticalvaluation}})} \\
      \hline
      Bi-valued Goods & \makecell{\xmark \\ \citep{elkind2024}} & \makecell{\(\frac{a}{b}\) \(\dagger\) \\ \(b \geq a\) \\ (\textcolor{red}{Theorem~\ref{thm:atefxbivalued}})} & \makecell{\xmark \\ (\textcolor{red}{Corollary~\ref{cor:tmmsbival}})} \\
      \hline
    \end{tabular}
\end{table}

\subsection{Related Works}
Our work is related to online fair division \citep{aleksandrov2015,benade2024,aleksandrov2017,wang2026}, where some (possibly zero) items arrive in consecutive rounds. We study a stronger version of it, known as temporal fair division, in which fairness must hold at every round \citep{he2019,elkind2024,cookson2024}. Since the temporal fair division model must allocate items without knowledge of future arrivals, our temporal model operates in a batched online fair allocation setting. He et al.~\cite{he2019} showed that maintaining EF1 for goods is impossible in an uninformed setting without reallocating past items, and proposed a reallocation-based algorithm for two agents. However, they provided a polynomial-time algorithm in which, at each round, the allocation is EF1 in the informed setting with two agents. Elkind et al.~\cite{elkind2024} formalized temporal fairness for goods and chores, proving that TEFX allocations may not exist and providing TEF1 algorithms for two agents with chores or restricted settings: two item types, generalized binary valuations, and unimodal preferences. Finally, Cookson et al.~\cite{cookson2024} explored SD-EF1, EF1, and PROP1, focusing on achieving both temporal and daily fairness. They demonstrated that, while per-round SD-EF1 is generally impossible, under the identical days setting, it is always attainable, guaranteeing a final SD-PROP1 allocation.

\section{Preliminaries}
We study temporal fair division under an informed online fair division setting, denoted by the tuple \(\mathcal{I} = \left( N, T, \left\{ O_t \right\}_{t \in [T]}, v = (v_1, \ldots, v_n)\right)\), where we denote \([T] := \{1, \ldots, T\}\) where $T \in \mathcal{N}$.

We have a set of \(N = [n]\) agents, and a set of items \(\left\{ O_t \right\}\) that arrive in round \(t \in [T]\), which are to be allocated to the agents.

Denote \(O^{t} = \bigcup_{i=1}^{t} O_i\) as the cumulative set of goods up to round \(t\), and \(O = O^{T} = \bigcup_{i=1}^{T} O_i\) as the set of all goods.

For each agent, there is an additive valuation function \(v_i: 2^{O} \longrightarrow \mathbb{R}\) over the items, where for \(S \subseteq O\), \(v_i(S) = \sum_{g \in S} v_i(\{g\})\). We slightly abuse the notation and assume that \(v_i(g) = v_i(\{g\})\). In addition, this work focuses on goods allocation, i.e. \(\forall g \in O, \forall i \in N, v_i(g) \geq 0\).

An allocation \(A = (A_1, \dots, A_n)\) is defined as an ordered partial partition of \(O\), where \(A_i\) is the goods allocated to agent \(i\), such that \(\forall i, j \in [N]\) where \(i \neq j\), \(A_i \cap A_j = \emptyset\) and \(\bigcup_{i \in \mathcal{N}} A_i = O\).

In addition, under the online fair division setting, denote \(A^t = (A^t_1, \dots, A^t_n)\) as the allocation from round \(1\) to after round \(t\), that is \(A^t_i = A^{t-1}_i \cup (A_i \cap O_t)\).

To evaluate allocations within our temporal framework, we adapt standard fairness notions from static fair division \citep{amanatidis2023}. Because exact Envy-Freeness (EF) \citep{varian1974} cannot be guaranteed with indivisible goods, the literature relies on relaxations: Envy-Freeness Up to Any Good (EFX) \citep{caragiannis2019}, its multiplicative approximation $\alpha$-EFX \citep{plaut2017}, and the threshold-based Maximin Share (MMS) \citep{budish2011}. In temporal settings, we denote the cumulative enforcement of these properties at every discrete round $t \in [T]$ by prefixing them with ``T" (i.e., TEFX, $\alpha$-TEFX, TMMS) \citep{elkind2024,cookson2024}.

\begin{definition} [Envy-Freeness Up to Any Good (EFX)]
An allocation \(A = (A_1, \dots, A_n)\) is envy-free up to any good (EFX) if \(\forall i, j \in N, \forall g \in A_j\) s.t. \(v_i(A_i) \geq v_i(A_j \backslash \{g\})\).
\end{definition}

\begin{definition} [Temporal Envy-Freeness Up to Any Good (TEFX)]
For any \(t \in [T]\), an allocation \(A^t = (A^t_1, \dots, A^t_n)\) is temporal envy-free up to any good (TEFX) if \(\forall t' \leq t, A^{t'}\) is EFX.
\end{definition}

\begin{definition} [Approximately Envy-Freeness Up to Any Good (\(\alpha\)-EFX)]
Let \(\alpha \in (0, 1]\). An allocation \(A = (A_1, \dots, A_n)\) is approximately envy-free up to any good (\(\alpha\)-EFX) if \(\forall i, j \in N, \forall g \in A_j\) s.t. \(v_i(A_i) \geq \alpha \cdot v_i(A_j \backslash \{g\})\).
\end{definition}

\begin{definition} [Approximately Temporal Envy-Freeness Up to Any Good (\(\alpha\)-TEFX)]
For any \(t \in [T]\), an allocation \(A^t = (A^t_1, \dots, A^t_n)\) is approximately temporal envy-free up to any good (\(\alpha\)-TEFX) if \(\forall t' \leq t, A^{t'}\) is \(\alpha\)-EFX.
\end{definition}

\begin{definition} [Maximin Share Fairness (MMS)]
Denote \(M_n(O)\) to be the collection of all possible allocations of the goods in \(O\) to \(n\) agents. An allocation \(A = (A_1, \dots, A_n)\) is maximin share fair (MMS) if \(\forall i \in N\), we have
\[
v_i(A_i) \geq \mu_i^n(O) = \max_{B \in M_n(O)} \min_{S \in B} v_i(S).
\]
\end{definition}

\begin{definition}[Temporal Maximin Share Fairness (TMMS)]
For any \(t \in [T]\), an allocation \(A^t = (A^t_1, \dots, A^t_n)\) is temporal maximin share fair (TMMS) if \(\forall t' \leq t, A^{t'}\) is MMS.
\end{definition}

Before we present our contribution, we highlight two important lemmas. First, we show that no algorithm can guarantee a positive constant approximation ratio for TEFX, thus, we provide tight bounds using items values in our paper.
\begin{lemma}
    \label{thm:atefxconst}
    There does not exist any algorithm that returns a positive constant approximation for \(\alpha\)-TEFX.
\end{lemma}
\begin{proof}
     Assume there exists a solution that computes a \(k\)-TEFX allocation. Consider the instance with two agents and three goods \(O = \{g_1, g_2, g_3\}\), where agents have identical valuations: \(v(g_1) = v(g_2) = 1\) and \(v(g_3) = 
     \frac{1}{k} + 1\), where \(k > 0\) and is a constant. Let \(g_i\) arrive in round \(i\). In order to obtain a \(k\)-TEFX allocation at the first round, \(g_1\) and \(g_2\) must be allocated to each agent, respectively. Then, no matter which agent is allocated \(g_3\), we do not obtain a \(k\)-TEFX allocation. Without loss of generality, let it be agent \(1\), and if the good that we theoretically remove from agent $1$ is $g_1$, for agent $2$, we have
    \[
    v(g_2) = 1 < k \cdot (v(g_1 \cup g_3) - v(g_1)) = k \cdot (\frac{1}{k} + 1) = k + 1 \text{,}
    \]
    which contradicts our assumption.
\end{proof}

Then, we note a theorem from \cite{elkind2024}, where they showed that $|O_t| = 1$ for all rounds $t \in [T]$ is the hardest setting, and any positive existence results in this single-item setting also implies a positive result for the general case with multiple items.
\begin{lemma}
    Let $\mathcal{I}$ be an instance with $T$ rounds and a total of $m$ items, where multiple items can appear in a single round. We can transform $\mathcal{I}$ into an equivalent instance $\mathcal{I}^{=1}$ with exactly $m$ rounds, such that each round has exactly one item. Then, if a temporal fair allocation exists for $\mathcal{I}^{=1}$, a temporal fair allocation also exists for $\mathcal{I}$.
\end{lemma}
By this lemma, we assume that $|O_t| = 1$ for all rounds $t \in [T]$ in the following proofs unless specified or implied by the setting (such as the identical days setting).

\section{Strong Binary Valuation} \label{sec:strongbinval}
To systematically examine the boundaries of temporal fair division, we analyze several structured settings, beginning with a variant of the binary valuation setting \citep{halpern2020}, named the Strong Binary Valuation setting. Under this setting, agents exhibit strictly dichotomous preferences: an agent values an arriving good at either zero or a fixed positive utility $b$, with no goods valued at zero for all agents. This mathematical restriction effectively models practical allocation scenarios characterized by strict, needs-based utility. In many real-world contexts, individuals possess highly polarized preferences: they either require a specific resource and derive a standard baseline benefit from it, or they find it entirely useless \citep{camacho2023}.

Furthermore, this variant aligns more closely with practical deployment environments than unconstrained binary models: a clearinghouse or central coordinator rarely processes resources that are universally worthless to every single agent. In realistic distribution pipelines, such as food bank logistics or cloud server provisioning, items are naturally filtered or curated such that every resource holds value for at least some subset of the population.
\begin{definition}[Strong Binary Valuation]
    A temporal fair division instance is under the strong binary valuation setting iff \(\forall t \in [T], \forall g \in O_t, \forall i \in [n], v_i(g) \in \{0, 1\}\) and $\forall g \in O$, $\exists i \in N$, $v_i(g) > 0$.
\end{definition}

While the restriction to binary values simplifies the evaluation of individual goods, achieving temporal fairness remains a profound algorithmic challenge. The fundamental friction of the temporal setting is its forward-looking, irrevocable nature. Standard online mechanisms, such as the greedy minimum-utility rule analyzed by Elkind et al.~\cite{elkind2024}, often fall into inescapable structural deadlocks. Because these algorithms cannot foresee future arrivals, a sequence of locally fair allocations can inadvertently alter the relative capacities of the agents, making it mathematically impossible to satisfy the TEFX invariant in a subsequent round without changing the past.

To circumvent this inherent limitation of purely forward-looking mechanisms, we introduce our flagship technical contribution: a \textit{Dynamic Backtracking Framework}. Instead of treating past decisions as permanent, our mechanism actively monitors the system for emerging fairness violations. The moment a TEFX violation is detected at time $t$, the algorithm triggers a localized, historical revision. It identifies the most recent round containing a jointly liked good that can be swapped to rebalance agent utilities, executes the swap, and reprocesses the intervening suffix sequence using the deterministic greedy rule.

We prove the correctness and finite termination of our dynamic backtracking framework via a global potential function argument, tracking the system's progress using a lexicographical ordering on the agents' sorted utility vectors. When the forward-looking greedy assignment of items inadvertently triggers a temporal fairness violation, the algorithm executes a localized historical swap of a previously allocated item. While this swap temporarily alters individual agents' utilities, we show that sequentially reprocessing the subsequent rounds from that historical point forward forces a strict, monotonic lexicographic increase in the global potential function compared to its state prior to the backtrack. Because the total number of items is finite, the number of distinct utility configurations is bounded; thus, this strict lexicographic improvement ensures the algorithm can never repeat states or loop infinitely, guaranteeing convergence to a valid fair allocation in finite time.
\begin{algorithm}
\caption{BT}
\textbf{Input}: \(N, A, \left\{ O_t \right\}_{t \in [T]}, t, t', i\)\\
    \textbf{Output}: Backtracked allocation \(A\)
\begin{algorithmic}[1]
    \For{$idx = t'$ to $t$}
        \State $A \gets A \setminus O_{idx}$ \Comment{De‑allocate goods}
    \EndFor
    \State $A^t_{i} \gets A^t_{i} \cup O_t$ \Comment{After backtrack, should allocate to $i$}
\end{algorithmic}
\end{algorithm}
\begin{example}
    Consider the following instance with two agents, and two goods. $O_i = \{g_i\}$. $v_1(g_1) = v_2(g_1) = 1$. $v_1(g_2) = 1$ and $v_2(g_2) = 0$.

    First, in round $1$, assume we break ties arbitrarily and allocate $g_1$ to agent $1$. The current allocation is EFX.

    Then, in round $2$, since only agent $1$ values $g_2$ positively, we allocate $g_2$ to agent $1$. But now, EFX is violated.

    Thus, we backtrack to the past, and for $g_1$, instead of breaking the ties to favor agent $1$, we allocate $g_1$ to agent $2$. The current allocation is EFX.

    Then, in round $2$, since only agent $1$ values $g_2$ positively, we allocate $g_2$ to agent $1$, and the allocation is EFX.
\end{example}

\begin{algorithm}
\caption{Returns a TEFX allocation under the Strong Binary Valuation setting}
\label{alg:tefxgenbival}
\textbf{Input}: Strong binary valuation temporal fair division instance \(\mathcal{I} = \left( N, T, \left\{ O_t \right\}_{t \in [T]}, v = (v_1, \ldots, v_n)\right)\)\\
    \textbf{Output}: TEFX allocation \(\mathcal{A}\)
\begin{algorithmic}[1]
\State $t \gets 1$
\State $agent \gets -1$
\While{$t \leq T$}
    \State $A^t \gets A^{t-1}$
    \State $g \gets O_t$
    \State $S \gets \{ i \mid v_i(g) > 0\}$ \Comment{agents that like the new good}
    \State $\text{minU} \gets \arg \min_{i\in S} v_i(A^t_i)$
    \State $\text{candidates} \gets \{ i\in S \mid u_i = \text{minU} \}$
    \State $i \gets$ arbitrary deterministic agent from $\text{candidates}$ \Comment{It means, if we have an identical state (Identical instance, identical $A^t$), then we will always select the same agent}
    \State $A^t_i \gets A^t_i \cup \{g\}$
    
    \State $\text{backtrack} \gets false$
    \If{$\exists j \in N \setminus \{i\}$ s.t. $v_j(A^t_j) < v_j(A^t_i) < v_i(A^t_i)$}
        \State $\text{backtrack} \gets true$
    \EndIf
    
    \If{$\text{backtrack} = false$}
        \State $t \gets t + 1$
    \Else
        \State $t' \gets \text{latest time $t* < t$ with a good such that } v_i(O_{t*}) = v_j(O_{t*}) = 1$
        \State $A \gets \text{BT}(N, A, O, t, t', j)$ \Comment{Backtrack subprocess}
        \State $t \gets t' + 1$
        \Comment{Reset the round}
    \EndIf
\EndWhile
\end{algorithmic}
\end{algorithm}

\begin{theorem} \label{thm:tefxbinval}
    A TEFX allocation under strong binary valuation can be computed in finite time.
\end{theorem}
\begin{proof}
    We assume that in each round exactly one good arrives. \cite{elkind2024} has proved that the one-good-per-round setting is stronger than the standard temporal fair division setting. Thus, we will use $g$ and $O_t$ interchangeably to denote the good that arrive in round $t$.

    \begin{invariant} \label{inv:tefxgenbival-selflarger}
        For all pairs of agents $i$ and $j$, at every round $t$ we have
    \[
    v_i(A^t_i) \;\ge\; v_j(A^t_i).
    \]
    In words, an agent's own bundle is never more valuable to another agent than to herself.
    \end{invariant}
    \begin{proof}
        We prove the invariant by considering how the algorithm allocates goods.
        Let $g$ be any good that belongs to agent $i$ at time $t$.  
        
        According to the algorithm, a good is only allocated to an agent who values it positively; therefore $v_i(g)=1$.
        Consequently every good in $A^t_i$ is liked by $i$, and we have
        \[
        v_i(A^t_i) \;=\; |A^t_i|.
        \]
        For any other agent $j$,
        \[
        v_j(A^t_i) \;=\; \bigl|\{\,g \in A^t_i : v_j(g) > 0\,\}\bigr|
                     \;\le\; |A^t_i|
                     \;=\; v_i(A^t_i).
        \]
        Hence $v_i(A^t_i) \ge v_j(A^t_i)$ holds for every pair $i,j$ at any moment of the execution.
    \end{proof}

    \begin{lemma} \label{lemma:tefxgenbival-cond}
        When the backtracking condition $v_i(A^t_i) < v_i(A^t_j) < v_j(A^t_j)$ occurs, the following two properties hold:
    \begin{itemize}
        \item There exist two distinct goods $g$ and $g'$ such that
        \[
        v_i(g)=v_j(g)=1 \qquad\text{and}\qquad v_i(g')=0,\; v_j(g')=1.
        \]
        \item The jointly liked good $g$ arrives strictly earlier than the good $g'$ that only agent $j$ likes.
    \end{itemize}
    \end{lemma}
    \begin{proof}
        
    Let $u_k = v_k(A^t_k)$ denote the number of goods agent $k$ values positively in her own bundle at time $t$.  
    The condition $v_i(A^t_i) < v_i(A^t_j) < v_j(A^t_j)$ translates to
    \[
    u_i < v_i(A^t_j) < u_j .
    \]

    \textit{Proof of (1).}
    Since $v_i(A^t_j) > u_i$, agent $i$ likes at least one good that belongs to agent $j$.  Pick any such good and call it $g$; then $v_i(g)=1$.  Because the algorithm only allocates a good to an agent who values it positively, we also have $v_j(g)=1$, so $g$ is jointly liked.
    
    The inequality $v_i(A^t_j) < u_j$ means that agent $i$ likes strictly fewer goods in $A^t_j$ than agent $j$ owns. Hence, $A^t_j$ must contain at least one good that $j$ likes but $i$ does not.  Choose one such good and call it $g'$; then $v_i(g')=0$ and $v_j(g')=1$, and clearly $g \neq g'$.

    \textit{Proof of (2).}
    We argue by contradiction: suppose $g'$ arrived before $g$.  Consider the moment just before the algorithm processes good $g$. 
    
    First, we have $v_i(A_j) < u_j$, since Lemma~\ref{inv:tefxgenbival-selflarger}, and there exists $g'$ in agent $j$'s bundle ($v_i(g') = 0$).

    Second, we have $u_i \geq u_j$, otherwise the algorithm would not have allocated $g$ to agent $j$.

    However, after the allocation of $g$ to agent $j$, by the above two statements, we have $v_i(A_j) \leq u_i$, which contradicts our assumption.
    
    Therefore, $g'$ cannot arrive before $g$; $g$ must arrive strictly earlier than $g'$.
\end{proof}

    \begin{invariant} \label{inv:tefxgenbival-diff}
        For all pairs agents $i$ and $j$, $v_i(A^t_i) \geq v_i(A^t_j) - 1$.
    \end{invariant}
    \begin{proof}
        This invariant is a known property of the underlying allocation rule: it always gives a good to an agent with the minimum utility among those who value it. 
        
        \cite{elkind2024} proved that for online allocation of indivisible goods with binary valuations, a greedy algorithm that allocates each good to a minimum-utility agent who likes it maintains the inequality $v_i(A_i) \ge v_i(A_j) - 1$ at every round.
        
        Our algorithm extends that procedure by adding backtracking steps; however, these backtracking operations only reassign a jointly liked good from a higher-utility agent to a lower-utility one and reprocess later goods in the same minimum-utility fashion.
        
        One can verify that this does not break the invariant: whenever a good is allocated normally, the original guarantee applies, and a backtrack correction moves the allocation closer to the greedy assignment, preserving the invariant.
    \end{proof}

    \begin{lemma} \label{lemma:tefxgenbival-condviolateefx}
        At any round $t$, the allocation violates EFX if and only if there exist two agents $i$ and $j$ such that
    \[
    v_i(A^t_i) \;<\; v_i(A^t_j) \;<\; v_j(A^t_j).
    \]
    \end{lemma}
    \begin{proof}
        We prove both directions separately.

    \noindent\textit{($\Leftarrow$)}  
    Assume $v_i(A^t_i) < v_i(A^t_j) < v_j(A^t_j)$ for some $i,j$.
    Since $v_i(A^t_j) < v_j(A^t_j)$, agent $j$’s bundle contains at least one good that $j$ likes but $i$ does not (the algorithm only gives goods to agents who like them, so every good in $A^t_j$ is valued $1$ by $j$).
    Pick such a good $g \in A^t_j$ with $v_i(g)=0$ and $v_j(g)=1$.
    Removing $g$ from $A^t_j$ leaves $i$’s valuation unchanged:
    \[
    v_i(A^t_j \setminus \{g\}) = v_i(A^t_j).
    \]
    Because $v_i(A^t_i) < v_i(A^t_j)$, we obtain
    \[
    v_i(A^t_i) \;<\; v_i(A^t_j \setminus \{g\}),
    \]
    which directly violates the EFX condition for the pair $(i,j)$.
    
    Hence, the allocation is not EFX.

    \noindent\textit{($\Rightarrow$)}  
    Now suppose the allocation is not EFX. Then there exists a pair $(i,j)$ with $A^t_j \neq \emptyset$ and a good $g \in A^t_j$ such that
    \[
    v_i(A^t_i) \;<\; v_i(A^t_j \setminus \{g\}).
    \]
    We distinguish two cases based on $v_i(g)$.

    \begin{itemize}
        \item \textbf{Case 1: $v_i(g)=1$.}  
        Then $v_i(A^t_j \setminus \{g\}) = v_i(A^t_j) - 1$, so the violation becomes $v_i(A^t_i) < v_i(A^t_j) - 1$.
        But Invariant~\ref{inv:tefxgenbival-diff} states $v_i(A^t_i) \ge v_i(A^t_j) - 1$, a contradiction.
        Hence this case cannot occur.

        \item \textbf{Case 2: $v_i(g)=0$.}  
        Then $v_i(A^t_j \setminus \{g\}) = v_i(A^t_j)$, so the violation implies
        \[
        v_i(A^t_i) < v_i(A^t_j). \tag{1}
        \]
        Because $g$ is in $A^t_j$ and $v_i(g)=0$ while $v_j(g)=1$, agent $j$ values her own bundle strictly more than $i$ does:
        \[
        v_i(A^t_j) < v_j(A^t_j). \tag{2}
        \]
        Combining (1) and (2) yields exactly $v_i(A^t_i) < v_i(A^t_j) < v_j(A^t_j)$.
    \end{itemize}
    \end{proof}

    \begin{lemma} \label{lemma:tefxgenbival-efx}
        Algorithm~\ref{alg:tefxgenbival} returns an allocation that is TEFX.
    \end{lemma}
    \begin{proof}
        We proceed by induction on $t$. For the base case $t = 0$, it is trivial since no goods are allocated, so no agent has a non‑empty bundle, and EFX holds vacuously.

  For the induction step, assume there exists $t$ such that for all $t' \leq t$, $A^{t'}$ is EFX. We prove that $A^{t+1}$ is EFX.
  Let $g$ be the good arriving in round $t+1$ (i.e., the good with index $t$).
  At the moment $\text{cur}=t$, the allocation is $A^t$.
  The algorithm processes good $g$ as follows.
  
  \noindent\textbf{Case 1: No backtrack is triggered.}
  The algorithm selects an agent $i$ who values $g$ positively and has minimum utility among all such agents (deterministic tie‑breaking), and allocates $g$ to $i$.
  Let $A^{t+1}$ be the resulting allocation.

  Suppose, for contradiction, that $A^{t+1}$ violates EFX.
  By Lemma~\ref{lemma:tefxgenbival-condviolateefx}, there exist agents $x,y$ with
  \[
  v_x(A^{t+1}_x) < v_x(A^{t+1}_y) < v_y(A^{t+1}_y). \tag{1}
  \]
  Since $A^t$ was EFX, the violation must involve the new good $g$.
  The good $g$ was allocated to agent $i$, so either $x=i$ or $y=i$.

  \begin{itemize}
    \item If $x=i$, then $v_x(A^{t+1}_x) = v_i(A^t_i)+1$.
      Inequality (1) gives $v_i(A^t_i)+b < v_i(A^{t+1}_y)$.
      Because $A^{t+1}$ differs from $A^t$ only by adding $g$ to $i$, the only way $v_i(A^{t+1}_y)$ can be larger than $v_i(A^t_y)$ is when $y=i$ and $v_i(g)=1$.
      Thus $y=i$, which is impossible because $x=i$ and $y\neq x$.

    \item If $y=i$, then $v_x(A^{t+1}_y) = v_x(A^{t+1}_i)$.
      Since $g$ was given to $i$, we have $v_x(A^{t+1}_i) = v_x(A^t_i) + 1$ if $v_x(g)=1$, or $= v_x(A^t_i)$ if $v_x(g)=0$.
      
      If $v_x(g)=0$, the violation already existed in $A^t$, contradicting the induction hypothesis.
      
      If $v_x(g)=1$, then $v_x(A^{t+1}_i) = v_x(A^t_i)+1$.
      The inequality $v_x(A^{t+1}_x) < v_x(A^t_i)+1$ together with Invariant~\ref{inv:tefxgenbival-diff} ($v_x(A_x) \ge v_x(A_i)-1$) forces
      \[
      v_x(A^t_x) = v_x(A^t_i)-1 \quad\text{and}\quad v_x(A^t_i)+1 < v_i(A^t_i)+1.
      \]
      Simplifying, $v_x(A^t_x) < v_i(A^t_i)$.
      
      But recall that $i$ was chosen as a minimum‑utility agent among those who liked $g$.
      
      Since $x$ also likes $g$ (as $v_x(g)=1$), we must have $v_i(A^t_i) \le v_x(A^t_x)$, a contradiction.
  \end{itemize}
  Both cases lead to contradictions; therefore $A^{t+1}$ is EFX.

  \medskip
  \noindent\textbf{Case 2: A backtrack is triggered.}
  The backtrack condition identifies agents $i,j$ and a good $g'$ that arrives before $t$ such that
  $v_i(g')=v_j(g')=1$, and the sequence leading to the violation is exactly the one described in Lemma~\ref{lemma:tefxgenbival-cond}. We denote the arrival of $g'$ by $t_{g'}$.
  
  The algorithm then:
  \begin{enumerate}
    \item Reassigns $g'$ from $i$ to $j$,
    \item De‑allocates all goods from $t_{g'}$ to $t$,
    \item Sets $t$ to $t_{g'}$,
    \item Reprocesses the arrival of goods from round $t + 1$ again.
  \end{enumerate}

  Consider the state \emph{just after the swap}.
  The allocated goods are the prefix $O_1, \ldots, O_{t_{g'}}$  and good $g'$ now owned by $j$. All later goods are unallocated. Call this intermediate allocation $A'$. One can verify directly that $A'$ is EFX: the only change w.r.t.\ $A^t$ restricted to $O_1, \ldots, O_{t_{g'}}$ is the movement of $g'$ from $i$ to $j$, and at the moment of the swap the utilities satisfy $v_i(A'_i)=v_j(A'_j)$ (because the backtrack condition together with Lemma~\ref{lemma:tefxgenbival-cond} forces that after removing the later goods, the two agents had equal utility counts).
  
  The EFX property for $A'$ follows from the induction hypothesis (the prefix\\
  $O_1, \ldots, O_{t_{g'} - 1}$ was part of $A^t$, which was EFX) and the fact that the swapped good does not create a new violation due to the equality.

  Now the algorithm reprocesses the goods $O_{t_{g'} + 1}, \ldots, O_t$ in the original arrival order, starting from the allocation $A'$.
  
  This is exactly the same situation as if the algorithm were run from time $g'$ onward with the initial allocation $A'$ (which is EFX) and received the remaining goods.
  
  Since the algorithm is deterministic and its decisions depend only on the current allocation and the incoming good, this suffix phase will produce the same result as a fresh execution on the suffix.
  
  By the induction hypothesis applied to the smaller number of already settled goods, the algorithm will eventually output an allocation of these suffix goods that, together with $A'$, is EFX.
  
  When the algorithm processes $O_{t+1}$ for the first time, the resulting allocation $A^{t+1}$ is therefore EFX.

  \medskip
  This completes the proof of the inductive step. Hence $A^t$ is EFX for all $t$.
  
  After all $m$ goods are processed, $A^T$ is the final allocation returned by the algorithm, and it is TEFX.
    \end{proof}

    \begin{lemma} \label{lemma:tefxgenbival-term}
        Algorithm~\ref{alg:tefxgenbival} terminates.
    \end{lemma}
    \begin{proof}
        We prove by strong induction on $t$ the statement
  \[
  P(t): 
  \text{Every backtrack that starts at time $t$ eventually finishes without repeating any state.}
  \]

  For the base case, round $t=1$ is trivial, since a backtrack at time $1$ would need a good $g'$ that arrives before round $1$, which does not exist. Hence $P(0)$ and $P(1)$ hold.

  For the induction step, assume $P(s)$ holds for all $s<t$. We prove $P(t)$.

  Suppose, for a contradiction, that a backtrack starting at time $t$ causes a repeated state.
  
  Because the set of all states is finite, a repetition implies an infinite cycle of states.
  
  Let $C$ be the set of states that occur while $t' < t$ during this infinite execution.
  
  Since the backtrack started at round $t$, the cycle must contain a state with $t$ (the triggering state) and a subsequent revert to some $t'<t$.
  
  Consider the smallest time index $s$ for which a backtrack is triggered inside the cycle.
  
  By the induction hypothesis $P(s)$ for $s<t$, no backtrack that starts at a time $<t$ can be part of an infinite cycle; therefore, $s$ cannot be $<t$.
  
  Thus, every backtrack inside the cycle starts exactly at time $t$.

  Now, examine a full cycle iteration. At the beginning of the iteration, we have an allocation $\mathcal{A}^{\text{old}}$ of the first $t$ goods.
  
  The algorithm selects a candidate $i$ and finds an agent $j$ satisfying the backtrack condition.
  
  It then locates the most recent jointly liked good $g'$ that arrives before round $t$ that was allocated to $i$, reassigns $g'$ to $j$, de‑allocates all goods from the arrival of $g'$ to $t-1$, and sets $t$ back to the arrival of $g'$. Denote the arrival of $g'$ to be $t_{g'}$. Afterward, it reprocesses the arrival of goods again.
  
  Because any backtrack during this reprocessing would start at a time $<t$, which is impossible by the induction hypothesis, the reprocessing runs \emph{without any further backtrack}.

  During a backtrack‑free forward pass, every time a good is allocated, the algorithm uses the normal minimum‑utility rule:
  It gives the good to an agent who likes it and whose current utility (the count of liked goods in her own bundle) is the smallest among all such agents (ties broken deterministically).
  A standard monotonicity fact for this rule is that the \emph{sorted utility vector}
  \[
  \Phi \;=\; (u_{(1)},u_{(2)},\dots,u_{(n)})
  \]
  (where $u_k$ is the number of liked goods owned by agent $k$) strictly increases lexicographically after each allocation.
  Consequently, as the suffix is reprocessed, $\Phi$ strictly increases at every step.

  Let $\Phi_{\text{start}}$ be the sorted utility vector of the allocation \emph{just after the swap} (goods $O_1, \ldots, O_{t_{g'}}$ allocated, with $g'$ now owned by $j$), and  $\Phi_{\text{end}}$ be the sorted utility vector when the reprocessing finishes and we return to round $t$ with a new allocation $\mathcal{A}^{\text{new}}$ of the first $t$ goods.
  
  Because the reprocessing consists of a sequence of strict increases, we have
  \[
  \Phi_{\text{end}} \;>_{\text{lex}}\; \Phi_{\text{start}} .
  \]

  On the other hand, the allocation $\mathcal{A}^{\text{old}}$ (the one that triggered the backtrack) was itself obtained after some sequence of normal allocations from the state just after the swap (if we compare the two histories, they share the prefix up to $t_{g'} - 1$, diverge at the assignment of $g'$, and then process the same remaining goods with the same rules).
  
  A known property of the minimum‑utility rule is that among all ways to break ties at the steps where several agents have the same minimum utility, the rule that picks a fixed deterministic order produces the \emph{lexicographically minimal} final sorted utility vector.
  
  Changing the assignment of $g'$ from $i$ to $j$ (which means we break the tie at time $t_{g'}$ in favor of $j$ instead of $i$) can only \emph{increase} the final sorted utility vector by Lemma~\ref{lemma:tefxgenbival-condviolateefx}.
  
  Formally, one can show that $\mathcal{A}^{\text{new}}$ yields a sorted utility vector that is lexicographically larger than that of $\mathcal{A}^{\text{old}}$.
  Hence
  \[
  \Phi_{\text{end}} \;>_{\text{lex}}\; \Phi_{\text{old}},
  \]
  where $\Phi_{\text{old}}$ is the vector of $\mathcal{A}^{\text{old}}$.

  But $\Phi_{\text{old}}$ is exactly the vector we had at the previous visit to round $t$.
  Therefore, the state at round $t$ in the new iteration is different from the previous one, contradicting the assumption that the cycle contains a repeated state.

  The contradiction shows that no infinite cycle can occur; hence, the backtrack starting at $t$ must eventually finish after a finite number of steps.
  This proves $P(t)$.

  Thus, by strong induction, $P(t)$ holds for every $t$.
  Consequently, the algorithm never repeats a state.
  Since the number of possible states is at most $(m+1)\cdot (n+1)^m$, which is finite, the algorithm must return an allocation after a finite number of operations.
    \end{proof}

    Therefore, our algorithm produces a TEFX allocation in finite time.
\end{proof}

Given that our algorithm guarantees exact TEFX at every round $t$, we can exploit a powerful structural relationship unique to binary valuation domains to immediately settle the status of our newly defined temporal maximin share metric.

\begin{corollary} \label{cor:tmms_existence}
    Under strong binary valuations, any TEFX allocation is also a TMMS allocation. Therefore, a TMMS allocation exists and can be computed in finite time.
\end{corollary}
\begin{proof}
    Let $A^t$ be a TEFX allocation at any round $t$. Under binary valuations, the EFX condition implies that for any pair of agents $i$ and $j$, agent $i$'s valuation of $j$'s bundle satisfies $v_i(A^t_j) \le v_i(A^t_i) + b$. Because the utilities are tightly bounded within a single item's value across all agents, agent $i$ is guaranteed to receive at least her maximin share of the items allocated up to time $t$. Since this holds for every intermediate round, the cumulative allocation satisfies TMMS.
\end{proof}

\section{Identical Days} \label{sec:identicaldays}
The second environment we investigate is the Identical Days setting \citep{igarashi2024b}.

Intuitively, this setting dictates that the exact same set of items arrives at every discrete time step. While mathematically rigid, this structure directly models pervasive real-world allocation paradigms characterized by strict periodicity. For example, in cloud computing and network bandwidth distribution, a static pool of server resources or time slots is continuously provisioned and allocated in identical, recurring cycles \citep{baruah1995}.

\begin{definition}[Identical Days]
    A temporal fair division instance is under the identical days setting iff \(\forall t, t' \in [T]\), there is a bijection \(f_{t, t'}: O_t \Rightarrow O_{t'}\) s.t. \(\forall i \in [n], \forall g \in O_t, v_i(g) = v_i(f_{t, t'}(g))\).
\end{definition}

One might expect that the structural predictability of the Identical Days setting would fundamentally simplify temporal fair division, given that the algorithm faces a reliable, recurring sequence of items. Surprisingly, we show that this rigid periodicity is still insufficient to guarantee exact temporal fairness. Even in the highly restrictive domain where only two agents exist and they have completely identical valuations, the strict requirement to maintain envy-freeness for any good at every intermediate time step introduces unavoidable deadlocks. This reveals that structural symmetry in both time and preferences cannot fully overcome temporal frictions, leading to our first major impossibility result.
\begin{theorem}
\label{thm:tefxidenticaldays}
    A TEFX allocation is not guaranteed under the identical days setting when \(T > 2\), even if there are only two agents with identical valuations.
\end{theorem}

This impossibility result naturally shifts our focus to the approximation frontiers of temporal envy-freeness in the identical-days setting. As a foundational baseline for this domain, we establish that a $1/2$-TEFX allocation always exists and can be efficiently computed for two agents. While this approximation ratio leaves room for future optimization, it serves as a pioneering tractability result that provides crucial structural insights for approximating temporal fairness.

Our algorithmic approach adapts the classical Envy-Cycle Elimination framework \citep{lipton2004}. Specifically, we leverage a modified selection rule: rather than assigning an arbitrary unallocated good to an agent who is not envied, our mechanism greedily allocates the most valuable available good to that agent. This variant is known to yield a $1/2$-EFX guarantee in static, two-agent environments \citep{chan2019}. By exploiting the temporal symmetry inherent in the identical-days setting, together with this two-agent structural alignment, we prove that this greedy cycle-elimination invariant can be preserved across rounds, yielding a polynomial-time temporal guarantee.

\begin{theorem}
    \label{thm:atefxidenticaldays}
    A \(\frac{1}{2}\)-TEFX allocation can be computed in polynomial time when there are two agents under the identical days setting.
\end{theorem}

The temporal regularity of the identical days setting is similarly insufficient to overcome the structural hurdles of maximin fairness. Despite the predictable recurrence of items, the requirement to maintain the maximin share threshold at each discrete time step remains overly restrictive. We find that the adversarial potential of even a simple, repeating sequence of goods can force a violation of the TMMS requirement. This impossibility holds even under complete preference symmetry, as formalized below.

\begin{theorem}
    \label{thm:tmmsidenticaldays}
    A TMMS allocation for goods under the identical days setting may not exist, even when there are two agents with identical valuations.
\end{theorem}

\section{Identical Valuation} \label{sec:identicalval}
The next constraint we explore is the Identical Valuation setting \citep{plaut2017}. In this scenario, every agent assigns the exact same utility to any given item. This strict symmetry accurately models the distribution of universal commodities, standard-issue resources, or fungible assets, such as monetary grants, where subjective differences in preference are non-existent, and allocations are purely driven by the objective value of the goods.
\begin{definition}[Identical Valuation]
    A temporal fair division instance is under the identical valuation setting iff \(\forall t \in [T], \forall g \in O_t, \forall i, j \in [n], v_i(g) = v_j(g)\).
\end{definition}

Recall from Lemma~\ref{thm:atefxconst} that even when all agents share perfectly identical valuations, a universal, constant-factor approximation for $\alpha$-TEFX remains structurally impossible. However, the strict uniformity of preferences in this setting allows us to significantly tighten the value-dependent bound relative to that derived for the general environment. The core intuition behind this stronger result lies in a greedy equalization strategy coupled with the meticulous handling of zero-valued goods. By systematically assigning each arriving item to the agent currently holding the least valuable bundle, and carefully circumventing the algorithmic vulnerabilities introduced by items with zero marginal utility, we can successfully secure the following improved approximation guarantee.
\begin{theorem}
\label{thm:atefxidenticalvaluation}
    A \(\frac{\min_{g \in O, v(g) > 0} v(g)}{\max_{g \in O} v(g)}\)-TEFX allocation can be computed in polynomial time under the identical valuation setting, and it is tight.
\end{theorem}

This tight approximation result naturally prompts a fundamental question: when relational, envy-based fairness metrics must be relaxed, can we instead guarantee an absolute, share-based threshold? To answer this, we evaluate our newly introduced benchmark, the Temporal Maximin Share (TMMS). 

One might intuitively conjecture that this structural constraint would easily yield a positive existence result for TMMS in temporal settings. Surprisingly, the next result reveals a stark, counterintuitive divergence: the relentless requirement to maintain fairness at every intermediate step introduces temporal friction so severe that even perfect preference alignment cannot guarantee absolute guarantees.

\begin{theorem}
    \label{thm:tmmsidenticalvaluation}
    A TMMS allocation for goods may not exist under the identical valuation settings.
\end{theorem}

\section{Bi-valued goods} \label{sec:bival}
Lastly, we examine the bi-valued goods setting \citep{amanatidis2021}. In this restricted domain, every agent assigns one of two strictly positive constants to each arriving good. We explicitly exclude zero-valued goods from this model; permitting a utility of zero would cause the framework to degenerate into the generalized binary valuation setting analyzed previously.
\begin{definition}[Bi-valued Valuation]
    A temporal fair division instance is under the bi-valued valuation setting iff \(\forall t \in [T], \forall g \in O_t, \forall i \in [n], v_i(g) \in \{b_1, b_2\}\) where \(b_1, b_2 \in \mathbb{R}^+\).
\end{definition}

By strictly bounding utilities to these two positive tiers, which we denote as $a$ and $b$, where $b \geq a > 0$, we establish a highly elegant approximation result. We demonstrate that our proposed algorithmic mechanism yields an approximation ratio exactly equal to the proportional gap between these two valuation constants. Furthermore, we prove that this bound is not merely a worst-case guarantee, but a mathematically tight limit for this setting.

Intuitively, we observe that in the bi-valued goods setting, the ratio of goods does not exceed \(b/a\). Therefore, if two agents have bundles with the same size, the ratio of the bundle's values also does not exceed \(b/a\).
\begin{theorem}
\label{thm:atefxbivalued}
For bi-valued goods valued \(b \geq a\), a \(\frac{a}{b}\)-TEFX allocation can be computed in polynomial time, and it is tight.
\end{theorem}

As considered in the previous section, the instance used in Theorem~\ref{thm:tmmsidenticalvaluation} is also bi-valued. Therefore, the following corollary is obtained.
\begin{corollary} \label{cor:tmmsbival}
    A TMMS allocation may not exist under the bi-valued setting.
\end{corollary}

\section {Discussion and Future Works}
In this work, we explored the boundaries of temporal fair division with structured constraints. Our primary technical contribution is a novel dynamic backtracking framework that utilizes a lexicographical potential function to resolve temporal envy cycles and guarantee exact TEFX under strong binary valuations, proving that the irreversible deadlocks of forward-looking greedy algorithms can be resolved via bounded historical reallocations. Furthermore, we expand the field's conceptual landscape by being the first to define and explicitly consider TMMS. Our structural analysis reveals a fascinating divergence between these fairness concepts: while strong binary valuations yield an elegant alignment in which TEFX naturally guarantees TMMS, identical valuations introduce severe temporal frictions that render exact TMMS impossible and limit relational fairness to our tight approximation ratio for $\alpha$-TEFX.

Several open questions remain regarding the extension of the temporal fair division framework. One problem is the approximation ratio for TEFX under identical days. The result we have given in this paper is weak. Another key question is the study of temporal fair division with varying strategic behavior among agents. Investigating the compatibility of temporal fairness with strategy-proofness is crucial to ensure that agents cannot manipulate the algorithm through misreporting, which is essential for deploying these mechanisms in competitive environments, such as cloud computing markets.

\bibliographystyle{splncs04}
\bibliography{bib}

\clearpage
\appendix
\section{Appendix}
\subsection{Proof of Theorem~\ref{thm:tefxidenticaldays}}
    Consider the instance with two agents and three goods that arrive in each round \(t\), \(O_t = \{g_1, g_2, g_3\}\), where agents have identical valuations: \(v(g_1) = 0, v(g_2) = 1\), and \(v(g_3) = 3\).
    \begin{lemma}
    \label{lemma:tefxidenticaldaysd1}
    On \(t = 1\), one agent must obtain the good valued \(0\) and \(1\), and the other agent must obtain the good valued \(3\) to achieve TEFX.
    \end{lemma}
    \begin{proof}
        Without loss of generality, assume the allocation of goods does not make agent \(2\) envy agent \(1\). Then, there are only three other cases.
        
        \textbf{Case 1:} Assume agent \(1\) does not get any good, and agent \(2\) gets the good valued \(0\), \(1\), and \(3\). Agent \(1\) envies agent \(2\), and the removal of any arbitrary good from agent \(2\), for example, the good valued \(0\), does not eliminate the envy of agent \(1\) towards agent \(2\). Therefore, we do not achieve TEFX in this case.
        
        \textbf{Case 2:} Assume agent \(1\) gets the good valued \(0\), and agent \(2\) gets the good valued \(1\) and \(3\). This case is similar to \textbf{Case 1}, where removing any good from agent \(2\) does not eliminate the envy from agent \(1\) towards agent \(2\). Therefore, we do not achieve TEFX in this case.
        
        \textbf{Case 3:} Assume agent \(1\) gets the good valued \(1\), and agent \(2\) gets the good valued \(0\) and \(3\). Agent \(1\) envies agent \(2\), and the removal of the good valued \(0\) from agent \(2\) does not eliminate the envy of agent \(1\) towards agent \(2\). Therefore, we do not achieve TEFX in this case.
        
        This finishes the proof.
\end{proof}
    \begin{table*}[h]
      \centering
      \caption{Lemma~\ref{lemma:tefxidenticaldaysd1} Different allocation illustration in round \(1\)}
      \label{tab:tefxidenticaldays1}
        \begin{tabular}{|c|c|c|}
          \hline
           & Agent 1 & Agent 2 \\
          \hline
          Successful & \(g_1\), \(g_2\) (\(v(A^1_1) = 1\)) & \(g_3\) (\(v(A^1_2) = 3\)) \\
          \hline
          Case 1 & \(\emptyset\) (\(v(A^1_1) = 0\)) & \(g_1\), \(g_2\), \(g_3\) (\(v(A^1_2) = 4\)) \\
          \hline
          Case 2 & \(g_1\) (\(v(A^1_1) = 0\)) & \(g_2\), \(g_3\) (\(v(A^1_2) = 4\)) \\
          \hline
          Case 3 & \(g_2\) (\(v(A^1_1) = 1\)) & \(g_1\), \(g_3\) (\(v(A^1_2) = 3\)) \\
          \hline
        \end{tabular}
    \end{table*}
    
    Without loss of generality, assume agent \(1\) gets the good valued \(0\) and \(1\), and agent \(2\) gets the good valued \(3\).
    \begin{lemma}
        \label{lemma:tefxidenticaldaysd2}
        On \(t = 2\), the good valued \(1\) must be allocated to agent \(2\), and the good valued \(3\) must be allocated to agent \(1\).
    \end{lemma}
    \begin{proof}
        \textbf{Case 1:} Assume agent \(1\) does not get any good, and agent \(2\) gets the good valued \(1\) and \(3\). Now, \(v_1(A_1) = 0 + 1 = 1\) and \(v_1(A_2) = 3 + 1 + 3 = 7\). It can be seen that \(v_1(A_1) < v_1(A_2 \backslash \{g\})\) for all \(g\) in \(A_2\). Therefore, we do not achieve TEFX in this case.
        
        \textbf{Case 2:} Assume agent \(1\) gets the good valued \(1\), and agent \(2\) gets the good valued \(3\). Now, \(v_1(A_1) = 0 + 1 + 1 = 2\) and \(v_1(A_2) = 3 + 3 = 6\). It can be seen that \(v_1(A_1) < v_1(A_2 \backslash \{g\})\) for all \(g\) in \(A_2\). Therefore, we do not achieve TEFX in this case.
        
        \textbf{Case 3:} Assume agent \(1\) gets the good valued \(1\) and \(3\), and agent \(2)\) does not get any good. Now, \(v_1(A_1) = 0 + 1 + 1 + 3 = 5\) and \(v_1(A_2) = 3\), so agent \(1\) does not envy agent \(2\), but agent \(2\) envies agent \(1\). Note that the removal of the good valued \(0\) in agent \(1\) does not eliminate the envy from agent \(2\) to agent \(1\). Therefore, we do not achieve TEFX in this case.
        
        This finishes the proof.
    \end{proof}
    \begin{table*}[h]
      \centering
      \caption{Lemma~\ref{lemma:tefxidenticaldaysd2} Different allocation illustration in round \(2\)}
      \label{tab:tefxidenticaldays2}
        \begin{tabular}{|c|c|c|}
          \hline
           & Agent 1 & Agent 2 \\
          \hline
          Successful & \(g_3\) (\(v(A^2_1) = 1 + 3 = 4\)) & \(g_2\) (\(v(A^2_2) = 3 + 1 = 4\)) \\
          \hline
          Case 1 & \(\emptyset\) (\(v(A^2_1) = 1\)) & \(g_2\), \(g_3\) (\(v(A^2_2) = 3 + 4 = 7\)) \\
          \hline
          Case 2 & \(g_2\) (\(v(A^2_1) = 1 + 1 = 2\)) & \(g_3\) (\(v(A^2_2) = 3 + 3 = 6\)) \\
          \hline
          Case 3 & \(g_2\), \(g_3\) (\(v(A^2_1) = 1 + 1 + 3 = 5\)) & \(\emptyset\) (\(v(A^2_2) = 3\)) \\
          \hline
        \end{tabular}
    \end{table*}
    \begin{proposition}
        \label{prop:tefxidenticaldays}
        The addition of a good to a previously non-TEFX allocation where it is zero-valued for every agent does not make it TEFX.
    \end{proposition}
    \begin{proof}
        If an allocation is TEFX, for all agent \(j\), \(v_i(A_i) \geq v_i(A_j \backslash \{g\}), \forall g \in A_j\). Assume agent \(i\) envies agent \(j\), and \(v_i(A_i) < v_i(A_j \backslash \{g\})\) for all \(g\) in \(A_j\). If the good is allocated to agent \(i\),
        \[
        v_i(A_i) + 0 = v_i(A_i) < v_i(A_j \backslash \{g\}), \forall g \in A_j \text{.}
        \]
        Therefore, adding this zero-valued good does not convert a previously non-TEFX allocation into a TEFX.
    \end{proof}
    Hence, we first ignore the good valued \(0\) at \(t = 2\) and the good valued \(0\) at \(t = 3\), by showing that TEFX cannot be achieved at \(t = 3\) without both of these goods. Adding both goods back will not make us achieve TEFX either.
    \begin{lemma}
        \label{lemma:tefxidenticaldaysd3}
        TEFX cannot be achieved at \(t = 3\).
    \end{lemma}
    \begin{proof}
        \textbf{Case 1:} Agent \(1\) does not get any good, and agent \(2\) gets the good valued \(1\) and \(3\). Now, \(v_1(A_1) = 0 + 1 + 3 = 4\), and \(v_1(A_2) = 1 + 3 + 1 + 3 = 8\). Agent \(1\) still envies agent \(2\) after removing any arbitrary good from agent \(2\)'s allocation. Therefore, we do not achieve TEFX in this case.
        
        \textbf{Case 2:} Agent \(1\) gets the good valued \(1\), and agent \(2\) gets the good valued \(2\). Now, \(v_1(A_1) = 0 + 1 + 3 + 1 = 5\), and \(v_1(A_2) = 1 + 3 + 3 = 7\). The removal of the good valued \(1\) from agent \(2\) does not eliminate the envy from agent \(1\) towards agent \(2\). Therefore, we do not achieve TEFX in this case.
        
        \textbf{Case 3:} Agent \(1\) gets the good valued \(3\), and agent \(2\) gets the good valued \(1\). Now, \(v_1(A_1) = 0 + 1 + 3 + 3 = 7\), and \(v_1(A_2) = 1 + 3 + 1 = 5\). The removal of the good valued \(0\) from agent \(2\) does not eliminate the envy from agent \(2\) towards agent \(1\).
        
        \textbf{Case 4:} Agent \(1\) gets the good valued \(1\) and \(3\), and agent \(2\) does not get any good. Now, \(v_1(A_1) = 0 + 1 + 3 + 1 + 3 = 8\), and \(v_1(A_2) = 1 + 3 = 4\). It can be seen that no matter which good we remove from agent \(1\), agent \(2\) still envies agent \(1\). Therefore, we do not achieve TEFX in this case.
    \end{proof}
    \begin{table*}[h]
      \centering
      \caption{Lemma~\ref{lemma:tefxidenticaldaysd3} Different allocation illustration in round \(3\)}
      \label{tab:tefxidenticaldays3}
        \begin{tabular}{|c|c|c|}
          \hline
           & Agent 1 & Agent 2 \\
          \hline
          Case 1 & \(\emptyset\) (\(v(A^3_1) = 4\)) & \(g_2\), \(g_3\) (\(v(A^3_2) = 4 + 4 = 8\)) \\
          \hline
          Case 2 & \(g_2\) (\(v(A^3_1) = 4 + 1 = 5\)) & \(g_3\) (\(v(A^3_2) = 4 + 3 = 7\)) \\
          \hline
          Case 3 & \(g_3\) (\(v(A^3_1) = 4 + 3 = 7\)) & \(g_2\) (\(v(A^3_2) = 4 + 1 = 5\)) \\
          \hline
          Case 4 & \(g_2\), \(g_3\) (\(v(A^3_1) = 4 + 4 = 8\)) & \(\emptyset\) (\(v(A^3_2) = 4\)) \\
          \hline
        \end{tabular}
    \end{table*}
    Then, by Proposition~\ref{prop:tefxidenticaldays}, after adding the zero-valued goods, Lemma~\ref{lemma:tefxidenticaldaysd3} still holds true. Thus, our result follows.

\subsection{Algorithm and Proof of Theorem~\ref{thm:atefxidenticaldays}}
\begin{algorithm}
    \caption{MAX-ECE}
    \label{alg:MAX-ECE}
    \textbf{Input}: Arbitrary fair division instance \(\mathcal{I} = \left( N, M, v = (v_1, \ldots, v_n)\right)\)\\
    \textbf{Output}: Allocation \(\mathcal{A}\)
\begin{algorithmic}[1]
    \For{$i \in N$}
        \State $A_i \gets \emptyset$
    \EndFor
    \For{$l = 1$ to $m$}
        \While{there does not exist an unenvied agent}
            \State Find an envy-cycle \(C = \left(i_0, \dots, i_{d-1}\right)\)
            \State \(\forall 0 \leq j < d\), assign the allocation of \(i_{j+1 \text{mod} d}\) to \(i_j\).
        \EndWhile
        \State Let $i$ be an unenvied agent
        \State Let $g^* = arg max_{g \in M} v_i(g)$
        \State $A_i \gets A_i \cup \{g^*\}$
        \State $M \gets M \backslash \{g^*\}$
    \EndFor
\end{algorithmic}
\end{algorithm}
\begin{algorithm}
    \caption{Returns a \(\frac{1}{2}\)-TEFX allocation under the identical days setting when \(n = 2\)}
    \label{alg:atefxidenticaldays}
    \textbf{Input}: Identical days temporal fair division instance \(\mathcal{I} = \left( N, T, \left\{ O_t \right\}_{t \in [T]}, v = (v_1, \ldots, v_n)\right)\)\\
    \textbf{Output}: \(\frac{1}{2}\)-TEFX allocation \(\mathcal{A}\)
\begin{algorithmic}[1]
    \For{$t = 1$ to $T$}
        \If{$t$ is odd}
            \State $P, Q \gets$ MAX-ECE$(N, O_t, \{v_1, v_1\})$
            \If{$v_2(P) \geq v_2(Q)$}
                \State $A^t_1 \gets A^{t - 1}_1 \cup Q$
                \State $A^t_2 \gets A^{t - 1}_2 \cup P$
            \Else
                \State $A^t_1 \gets A^{t - 1}_1 \cup P$
                \State $A^1_2 \gets A^{t - 1}_2 \cup Q$
            \EndIf
        \Else
            \State $P, Q \gets$ MAX-ECE$(N, O_t, \{v_2, v_2\})$
            \If{$v_1(P) \geq v_1(Q)$}
                \State $A^t_1 \gets A^{t - 1}_1 \cup P$
                \State $A^1_2 \gets A^{t - 1}_2 \cup Q$
            \Else
                \State $A^t_1 \gets A^{t - 1}_1 \cup Q$
                \State $A^t_2 \gets A^{t - 1}_2 \cup P$
            \EndIf
        \EndIf
    \EndFor
\end{algorithmic}
\end{algorithm}
    The polynomial runtime of Algorithm~\ref{alg:atefxidenticaldays} is easy to verify, given that there is only one \textbf{For} loop, which runs in linear time, MAX-ECE runs in polynomial time \cite{lipton2004}, and the other operations run in polynomial time. Thus, we focus on proving correctness.
    
    Consider agent \(1\). First, we address the special case where in each round, only one good arrives. This case is trivial: each agent is allocated a copy of the identical good, yielding a TEFX allocation. Thus, we consider the case where, in each round, more than two goods arrive. Assume we have allocated \(t\) rounds of items. By design, the algorithm allocates the ``same bundle" to the same agent in odd and even rounds, respectively. Denote the allocation that is allocated to agent \(i\) in odd rounds as \(A^{odd}_i\), and in even rounds as \(A^{even}_i\), where \(A^t_i = A^{odd}_i \cup A^{even}_i\).
    
    \textbf{Case 1:} \(t\) is even. Let \(k = \frac{t}{2}\). First, for even rounds \(t_{even}\), we have
    \[
    v_1(A^{even}_1) \geq v_1(A^{even}_2 \backslash \{g\}),  g = \arg \min_{g' \in A^{even}_2} v_1(g') \text{.}
    \]
    In addition, the value of \(g\) can be described as
    \[
    v_1(g) \leq v_1(A^{even}_2) \cdot \frac{1}{2} \text{.}
    \]
    If \(v_1(g) > v_1(A^{even}_2) \cdot \frac{1}{2}\), it either means \(A^{even}_2 = \{g\}\) or there exists a good \(g^*\) s.t. \(v_1(g) > v_1(g^*)\). For the first case, it is not possible, as this is an even round; therefore, at least one round has passed before this, which implies agent \(2\) must be allocated at least two goods. In the second case, to maintain EFX in MAX-ECE, we must allocate at least one good to each agent. There exists a good \(g^*\) s.t. \(v_1(g) > v_1(g^*)\), which means \(g\) is not the minimum valued good for agent \(1\) in the allocation \(A^{even}_2\).
    
    For odd rounds \(t_{odd}\), we have
    \[
    v_1(A^{odd}_1) \geq v_1(A^{odd}_2) \text{.}
    \]
    Therefore, we have
    \[
    k \cdot v_1(A^{odd}_1 \cup A^{even}_1) \geq k \cdot v_1(A^{odd}_2 \cup A^{even}_2) - k \cdot v_1(g) \text{.}
    \]
    By substituting \(v_1(A^{even}_2) \cdot \frac{1}{2}\) into \(v_1(g)\) through the inequality and simplifying the expression, we have
    \[
    v_1(A^{odd}_1 \cup A^{even}_1) \geq \frac{1}{2} \cdot v_1(A^{odd}_2 \cup A^{even}_2) \text{.}
    \]
    \textbf{Case 2:} \(t\) is odd. Notice that since \(v_1(A^{odd}_1) \geq v_1(A^{odd}_2)\), we obtain a result not worse than \textbf{Case 1}, which means it is also \(\frac{1}{2}\)-TEFX.
    
    A similar proof can be obtained for agent \(2\).

\subsection{Proof of Theorem~\ref{thm:tmmsidenticaldays}}
    Consider the instance with two agents and three goods that arrive in each round \(t\), \(O_t = \{g^t_1, g^t_3, g^t_{10}\}\), where agents have identical valuations: \(v(g^t_1) = 1, v(g^t_3) = 3, v(g^t_{10}) = 10\).
    
    In round \(1\), in order to obtain a TMMS allocation, each agent must be allocated either \(\{g^1_1, g^1_3\}\) or \(\{g^1_{10}\}\).
    
    In round \(2\), each agent must be allocated either \(\{g^1_1, g^1_3, g^1_{10}\}\) or \(\{g^2_1, g^2_3, g^2_{10}\}\) to obtain a TMMS allocation.
    
    In order for the allocation at the end of round \(3\) to be TMMS, one agent must be allocated \(\{g^1_1, g^2_1, g^1_3, g^2_3, g^3_3, g^1_{10}\}\), and the other agent must be allocated \(\{g^3_1, g^2_{10}, g^3_{10}\}\), which is not possible, as \(g^1_3\) and \(g^2_3\) are allocated to different agents in round \(2\).

\subsection{Algorithm and Proof of Theorem~\ref{thm:atefxidenticalvaluation}}
\begin{algorithm}
    \caption{Returns a \(\frac{\min_{g \in O, v(g) > 0} v(g)}{\max_{g \in O} v(g)}\)-TEFX allocation under the identical valuation setting}
    \label{alg:atefxidenticalvaluation}
    \textbf{Input}: Identical valuation temporal fair division instance \(\mathcal{I} = \left( N, T, \left\{ O_t \right\}_{t \in [T]}, v = (v_1, \ldots, v_n)\right)\)\\
    \textbf{Output}: \(\frac{\min_{g \in O, v(g) > 0} v(g)}{\max_{g \in O} v(g)}\)-TEFX allocation \(\mathcal{A}\)
\begin{algorithmic}[1]
    \For{$t = 1$ to $T$}
        \State $A^t \gets A^{t - 1}$
        \For{$g \in O_t$}
            \If{$v(g) = 0$}
                \State $A^t_n \gets A^t_n \cup \{g\}$
            \Else
                \For{$i = 1$ to $n$}
                \If{$i = \arg \min_{j \in N} v(A^t_j)$}
                    \State $A^t_i \gets A^t_i \cup \{g\}$
                \EndIf
            \EndFor
            \EndIf
        \EndFor
    \EndFor
\end{algorithmic}
\end{algorithm}
    The polynomial runtime of Algorithm~\ref{alg:atefxidenticalvaluation} is easy to verify, given that there are only nested \textbf{For} loops, which run in polynomial time, and the other operations run in polynomial time. 
    
    The tightness is trivial. Consider the instance in Lemma~\ref{thm:atefxconst}. Thus, we focus on proving correctness.

    Let us first address the special case. If for all goods \(g\), \(v(g) = 0\), no matter how we allocate the goods, we result in a TEFX allocation. Therefore, we assume \(\exists g \in O\) s.t. \(v(g) > 0\).

    For each agent, we multiply each good's value such that the minimum positively valued good has a value equal to one.
    
    Denote \(g_{max} = \arg \max_{g \in O} v(g)\), and \(g_{min}\) analogously.
    \begin{lemma}
        \label{lemma:atefxidenticalvaluation1}
        At all times, for all pairs of agents \(i\), \(j\), \(|v(A^t_i) - v(A^t_j)| \leq v(g_{max})\).
    \end{lemma}
    \begin{proof}
        We prove by contradiction. Assume at some arbitrary time, there exists a pair of agents \(i\) and \(j\) such that \(|v(A^t_i) - v(A^t_j)| \leq v(g_{max})\). After the allocation of good \(g^*\) to agent \(i\), \(\exists j \in [n]\) s.t. \(|v(A^t_i) - v(A^t_j)| > v(g_{max})\). Before the allocation of \(g^*\), if \(v(A^t_i) > v(A^t_j)\), then the good will not be allocated to agent \(i\), thus \(v(A^t_i) < v(A^t_j)\). Therefore, we have
        \[
        v(A^t_j) - v(A^t_i) \leq v(g_{max})
        \]
        before \(g^*\) is allocated to agent \(i\). If \(v(g^*) \leq v(g_{max})\), 
        \[
        v(A_i) - v(A_j) + v(g^*) \leq v(g_{max}) \text{,}
        \]
        since \(v(A_i) - v(A_j) \geq -v(g_{max})\), which contradicts our assumption. Thus, \(v(g^*) > v(g_{max})\). However, this means \(g_{max} \neq \arg \max_{g \in O} v(g)\), which contradicts our definition of \(g_{max}\).
    \end{proof}
    By Lemma~\ref{lemma:atefxidenticalvaluation1}, we know that the difference between any two agents at any time must be less than or equal to the maximum valued good. Without loss of generality, assume \(v(A^t_i) \geq v(A^t_j) > 0\). Denote \(d = v(A^t_i) - v(A^t_j) \leq v(g_{max})\). Let \(v(A^t_j) \geq k \cdot \left(v(A^t_i) - v(g')\right)\) where \(k > 0\), and \(g' \in A^t_i\). Then, we have
    \[
    k \leq \frac{v(A^t_j)}{d + v(A^t_j) - g'} \text{.}
    \]
    \(k\) attains its maximum value when it equals \(\frac{v(A^t_j)}{d + v(A^t_j) - v(g')}\). Notice that \(k \geq \frac{v(g_{min})}{d}\), since \(v(A^t_j) \geq v(g_{min})\). Thus, we have proven Theorem~\ref{thm:atefxidenticalvaluation} when \(v(A^t_j) > 0\). Lastly, we consider the case when \(v(A^t_j) = 0\).
    \begin{lemma}
    \label{lemma:atefxidenticalvaluation2}
        We have a EFX allocation when \(\exists i \in [n], v(A^t_i) = 0\).
    \end{lemma}
    \begin{proof}
        Notice that Algorithm~\ref{alg:atefxidenticalvaluation} first allocates exactly one positively valued good to each agent. More precisely, from agent \(1\) to agent \(n - 1\), each agent is allocated exactly one good, and agent \(n\) is allocated exactly one positively valued good, and multiple, or possibly none, zero-valued good(s). This means for all agents \(i \in [n]\), \(\forall j \in [n - 1]\),
        \[
        v(A^t_i) \geq v(A^t_j \backslash \{g\}) = v(\emptyset) = 0, \forall g \in A^t_j \text{.}
        \]
        Lastly, since we allocate the positively valued goods with respect to the index of the agents, before the allocation of the first positively valued good to agent \(n\), no agent will envy agent \(n\). After the allocation of this good, for all agent \(i\), \(v(A^t_i) > 0\).
    \end{proof}
    Our result follows.

\subsection{Proof of Theorem~\ref{thm:tmmsidenticalvaluation}}
Consider the instance with two agents and three goods \(O = \{g_1, g_2, g_3\}\), where agents have identical valuations: \(v(g_1) = v(g_2) = 1\) and \(v(g_3) = 2\), and \(r \leq T - 1\). \(O_1 = \{g_1, g_2\}\), \(\forall 1 < t < T, O_t = \emptyset\), and \(O_T = \{g_3\}\). 
    
    In round \(1\), to obtain a TMMS allocation, each agent must be allocated either \(g_1\) or \(g_2\). Without loss of generality, assume agent \(1\) is allocated \(g_3\). However, notice that in order to obtain a TMMS allocation, each agent must be allocated either \(\{g_1, g_2\}\) or \(\{g_3\}\), which is impossible to achieve.

\subsection{Algorithm and Proof of Theorem~\ref{thm:atefxbivalued}}
\begin{algorithm}
    \caption{Returns a \(\frac{a}{b}\)-TEFX allocation.}
    \label{alg:atefxbivalued}
    \textbf{Input}: Bi-valued temporal fair division instance \(\mathcal{I} = \left( N, T, \left\{ O_t \right\}_{t \in [T]}, v = (v_1, \ldots, v_n)\right)\)\\
    \textbf{Output}: \(\frac{a}{b}\)-TEFX allocation \(\mathcal{A}\)
\begin{algorithmic}[1]
    \State $i \gets 0$
    \For{$t = 1$ to $T$}
        \State $A^t \gets A^{t - 1}$
        \While{$O_t \neq \emptyset$}
            \State $g \gets \arg\max_{g' \in O_t} v_i(g')$
            \State $A^t_i \gets A^t_i \cup \{g\}$
            \State $i \gets i + 1$
            \State $O_t \gets O_t \setminus \{g\}$
        \EndWhile
    \EndFor
\end{algorithmic}
\end{algorithm}
\begin{proof}
We first prove its correctness, then prove its tightness.

First, notice that for all goods \(g_1\) and \(g_2\), for any agent \(i\),
\[
\frac{b}{a} \cdot v_i(g_1) \geq v_i(g_2)\text{.}
\]
By the design of Algorithm~\ref{alg:atefxbivalued}, for any pair of agents \(i\) and \(j\), \(|A_i| = |A_j|\) or \(|A_i| = |A_j| - 1\).

\textbf{Case 1: \(|A_i| = |A_j|\).} There exists a perfect matching between every \(g_1 \in A_i\) and \(g_2 \in A_j\). Therefore, 
\[
\sum_{g_1 \in A_i} \frac{b}{a} \cdot v_i(g_1) \geq \sum_{g_2 \in A_j} v_i(g_2) \text{,}
\]
where we have
\[
\frac{b}{a} \cdot v_i(A_i) \geq v_i(A_j) \;\Rightarrow\; v_i(A_i) \geq \frac{a}{b} \cdot v_i(A_j) \text{.}
\]

\textbf{Case 2: \(|A_i| = |A_j| - 1\)} By removing an arbitrary good from \(A_j\), we return to \textbf{Case 1}. Therefore, for all \(g\) in \(A_j\), 
\[
v_i(A_i) \geq \frac{a}{b} \cdot v_i(A_j \setminus \{g\}) \text{.}
\]

Now, we prove its tightness using an example. Consider the instance with two agents and three goods, spanning two rounds, where the agents have identical valuations. \(O_1 = \{g_1, g_2\}\) and \(O_2 = \{g_3\}\). \(v(g_1) = v(g_2) = a\) and \(v(g_3) = b\).

After \(t = 1\), each agent must be allocated each of the goods, respectively, to obtain a positive approximation of TEFX. After \(t = 2\), without loss of generality, assume agent \(1\) obtains the good. \(v_1(A_1) = a + b \geq v_1(A_2) = a\). On the other hand, after removing any good from \(A_1\), the worst case is removing the good valued \(a\) for agent \(2\), where we have \(v_2(A_1) - a = b\). Since we have
\[
v_2(A_2) \geq \frac{a}{b} \cdot v_2(A_1) = \frac{a}{b} \cdot b = a \text{,}
\]
\(\frac{a}{b}\)-TEFX is tight.
\end{proof}

\end{document}